\documentclass[AMA,Times1COL]{WileyNJDv5} 

\articletype{Research Article}%

\received{Date Month Year}
\revised{Date Month Year}
\accepted{Date Month Year}
\journal{Journal}
\volume{00}
\copyyear{2023}
\startpage{1}

\newcommand\independent{\protect\mathpalette{\protect\independenT}{\perp}}
\def\independenT#1#2{\mathrel{\rlap{$#1#2$}\mkern2mu{#1#2}}}

\newcommand{\var}{\mbox{var}}
\newcommand{\cov}{\mbox{cov}}
\newcommand{\calV}{\mathcal{V}}
\newcommand{\inD}{\stackrel{D}{\longrightarrow}}
\newcommand{\simapprox}{\stackrel{\cdot}{\sim}}

\begin{document}

\title{Independent increments and group sequential tests}

\author[1]{Anastasios A. Tsiatis}

\author[1]{Marie Davidian}


\authormark{TSIATIS and DAVIDIAN}
\titlemark{Independent increments and group sequential tests}

\address[1]{\orgdiv{Department of Statistics}, \orgname{North Carolina State University}, \orgaddress{\state{North Carolina}, \country{USA}}}



\corres{Corresponding author {Marie Davidian,
  \email{davidian@ncsu.edu}}}

\presentaddress{Department of Statistics, North Carolina State
  University, Raleigh, North Carolina, USA..}


\abstract[Abstract]{Widely used methods and software for group
  sequential tests of a null hypothesis of no treatment difference
  that allow for early stopping of a clinical trial depend primarily
  on the fact that sequentially-computed test statistics have the
  independent increments property. However, there are many practical
  situations where the sequentially-computed test statistics do not
  possess this property.  Key examples are in trials where the primary
  outcome is a time to an event but where the assumption of
  proportional hazards is likely violated, motivating consideration of
  treatment effects such as the difference in restricted mean survival
  time or the use of approaches that are alternatives to the familiar
  logrank test, in which case the associated test statistics may not
  possess independent increments.  We show that, regardless of the
  covariance structure of sequentially-computed test statistics, one
  can always derive linear combinations of these test statistics
  sequentially that do have the independent increments property. We
  also describe how to best choose these linear combinations to target
  specific alternative hypotheses, such as proportional or
  non-proportional hazards or log odds alternatives.  We thus derive
  new, sequentially-computed test statistics that not only have the
  independent increments property, supporting straightforward use of
  existing methods and software, but that also have greater power
  against target alternative hypotheses than do procedures based on
  the original test statistics, regardless of whether or not the original
  statistics have the independent increments property.  We illustrate
  with two examples.}



\maketitle

\renewcommand\thefootnote{}

\renewcommand\thefootnote{\fnsymbol{footnote}}
\setcounter{footnote}{1}

\section{Introduction}\label{s:intro}

Interim monitoring of clinical trials to allow for the possibility of
early stopping for efficacy or futility is a critically important
practice in health sciences and biopharmaceutical research and is
ordinarily accomplished through the use of group sequential methods.
Key to straightforward implementation of such methods is the property
of independent increments of sequentially-computed test statistics;
see Kim and Tsiatis\cite{Kim} for a comprehensive overview.
Specifically, computation of sequential stopping boundaries to which
the test statistic is compared at interim analyses requires
multivariate integration with respect to the (usually asymptotic)
joint distribution of the test statistic across monitoring times.
However, if the increments between statistics used to construct the
test statistics are independent, the computation simplifies to involve
only simple univariate integration, which is the foundation for the
widely used group sequential methods \cite{Pocock,OBF,LanDeMets} that
are available in standard software.

In many trials, inference focuses on a treatment effect that can be
characterized by a single parameter in a relevant statistical model,
e.g., the difference of treatment means in the case of a continuous
outcome or the risk difference, risk ratio, or odds ratio with a
binary outcome, where the definition of the parameter is the same
across all time.  Under these conditions, it has been shown
\cite{Scharfstein, Jennison} that efficient estimators of such
treatment effect parameters, computed sequentially, have the
independent increments property, so that standard group sequential
methods are readily implemented.  However, there are many situations
where efficient test procedures may be unavailable or not
straightforward to implement.  For example, Van Lancker et
al. \cite{Lancker} discuss use of a test statistic based on an
estimator for the treatment effect parameter that incorporates
covariate adjustment for prespecified baseline covariates thought to
be prognostic for the outcome.  Although covariate adjustment can
enhance the precision of the final analysis, owing to the complexity
of the efficient estimator,  the estimator and thus statistic the authors
adopt for practical use is not efficient and thus does not enjoy
the independent increment structure, so that standard group sequential
methods cannot be used.

The independent increments property also may not hold because of the
nature of the treatment effect.  In clinical trials in chronic
disease and especially in cancer, the primary outcome is a time to an
event, and the most common characterization of treatment effect is
through the hazard ratio under the assumption of proportional hazards,
with the primary analysis typically based on the logrank test or Cox
model. However, the proportional hazards assumption is often violated,
which has led to interest in representing the treatment effect through
the difference in restricted mean survival time (RMST), the mean time
to the event restricted to a specified time $L$.  As discussed by
Murray and Tsiatis \cite{Murray} and Lu and Tian \cite{LuTian},
because of limited follow up, at the $j$th interim analysis at time
$t_j$, say, the parameter of interest is the difference in RMST to a
time $L_j \leq t_j$; thus, the parameter of interest changes over time,
and it can be shown that the independent increments property does not
hold for the corresponding statistic.

Test statistics that are motivated by intuitive considerations may or
may not possess the independent increments property.  For example, an
alternative approach when violation of the proportional hazards
assumption is thought to be due to differences that manifest early in
time is to base the analysis on Gehan's Wilcoxon test \cite{Gehan},
which gives more weight to events at early time points.  Here, the
null hypothesis of a difference in event time distributions cannot be
expressed in terms of a single parameter, and the statistic on which
this test is based has been shown to not possess the independent
increments property\cite{SludandWei}.  In contrast, that for the logrank test,
which can also be motivated from an intuitive perspective, does have
independent increments.

In this article, we show that, regardless of the covariance structure
of sequentially-computed statistics, and thus regardless of whether or
not they possess the independent increments property, it is always
possible to derive sequential linear combinations of these statistics
that do possess this property.  Thus, for a given statistic, a group
sequential test procedure based on these linear combinations can be
implemented using standard methods and software.  If the original
statistics are not efficient, our formulation shows that linear
combinations of them leading to independent increments can be chosen
judiciously to also result in more efficient tests.  As a result,
depending on the alternative of interest, basing interim monitoring on
tests involving such linear combinations of statistics may yield
stronger evidence for early stopping, even if the original statistics
possess the independent increments property.  In
Section~\ref{s:model}, we present the general statistical framework
and the main result. In Section~\ref{s:implement}, we show how the
result is used to choose linear combinations to yield independent
increments and increase efficiency.  We consider two examples in
Section~\ref{s:examples}: interim monitoring based on Gehan's Wilcoxon
test when the treatment effect is the difference in treatment-specific
survival distributions, and interim monitoring when the treatment
effect is the difference in treatment-specific RMST.

\section{Statistical framework and main result}\label{s:model}

The property of independent increments can be characterized formally
as follows.  Consider an arbitrary random vector
$Z = (Z_1,\ldots ,Z_K)^T$ with $(K \times K)$ covariance matrix
$\calV_Z$.  Then $Z$ has the independent increments property if
${ \mathcal{V}}_Z(j,k)={\mathcal{V}}_Z(j,j)$, $j\le k$,
$j,k=1,\ldots ,K$, where ${\mathcal{V}}_Z(j,k)$ denotes the $(j,k)$th
element of $\calV_Z$.

We now state and prove the general result that serves as the
foundation for the methodology.  To this end, let $X_1,\ldots, X_K$
denote time ordered, sequentially-computed random variables.  Later,
these variables will be (suitably normalized) statistics computed at
$K$ possible interim analyses of a clinical trial using all the
available data through each interim analysis.  Often, such 
statistics are asymptotically normally distributed and constructed so
that they have mean zero under the relevant null hypothesis of no
treatment effect.  Let $X = (X_1,\ldots ,X_K)^T$ be the $K$ random
variables collected into a vector, and, for $j = 1,\ldots,K$, define
the $(K \times 1)$ vector
$\underline{X}_j = (X_1,\ldots ,X_j,0,\ldots ,0)^T$; that is, the
$K$-dimensional vector leading with $(X_1,\ldots ,X_j)^T$ followed by
$(K-j)$ zeros. Denote the $(K \times K)$ covariance matrix of $X$ by
${\mathcal{V}}$ and the covariance matrix of $\underline{X}_j$ by
${\mathcal{V}}_j$, the $(K \times K)$ matrix with upper left
$(j \times j)$ submatrix equal to the covariance matrix of
$(X_1,\ldots ,X_j)^T$ and all remaining elements equal to zero.

For $j = 1,\ldots,K$, let $a_j=(a_{j1},\ldots ,a_{jj},0,\ldots ,0)^T$
be a $K$-dimensional vector, where the first $j$ components
$a_{j1},\ldots ,a_{jj}$ are constants followed by $(K-j)$ zeros.  Then
denote the sequentially-computed linear combinations of $X$ by
\begin{equation}\label{eq.1}
Y=(Y_1,\ldots ,Y_K)^T, \hspace{0.15in}  Y_1=a_1^TX,\ldots , Y_K=a_K^TX.
\end{equation}
Henceforth, assume that the covariance matrix $\mathcal{V}$ has full
rank and thus a unique inverse $\calV^{-1}$ and that the linear
combinations in (\ref{eq.1}) are nontrivial in the sense that
$Y_j=a_{j1}X_1+\ldots +a_{jj}X_j$, which in the context of interim
monitoring would be the linear combination computed at the $j$th
interim analysis, must include $X_j$; i.e., $a_{jj}\ne 0$,
$j=1,\ldots ,K$. This condition guarantees that the covariance matrix
of $Y$ also has full rank.  The key result, that these
sequentially-computed linear combinations of $X$ have the independent
increments property, follows from Theorem~\ref{thm1}.

\begin{theorem} \label{thm1} Under the above conditions, with $\calV$
  of full rank and $\calV_j$, $j = 1,\ldots,K$, defined as above, the
  vector of nontrivial, sequentially-computed linear combinations
  $Y=(Y_1,\ldots ,Y_K)^T$ has the independent increments property if
  and only if there exists a vector of constants
  $b=(b_1,\ldots ,b_K)^T$ such that
  $a_j=\mathcal{V}_j^{-} \underline{b}_j$, where
  $\underline{b}_j=(b_1,\ldots ,b_j,0,\ldots ,0)^T$, $j=1,\ldots ,K$,
  and $\calV_j^{-}$ is the $(K\times K)$ matrix with upper
  left hand $(j \times j)$ submatrix the inverse of the covariance
  matrix of $(X_1,\ldots ,X_j)^T$ and all remaining elements of the
  matrix equal to zero; i.e., $\calV_j^{-}$ is the
  Moore-Penrose generalized inverse of $\calV_j$.
  \end{theorem}

  \begin{proof}
First, assume that $a_j=\calV_j^{-} \underline{b}_j$ holds.  It is
straightforward that, in general,
${\mathcal{V}}_Y(j,k)=a_j^T{\mathcal{V}}a_k$, $j,k=1,\ldots ,K$;
thus, it follows that ${\mathcal{V}}_Y(j,k) =
\underline{b}_j^T\calV_j^{-}{\mathcal{V}}\calV_k^{-}\underline{b}_k$.  
However, for $j\le k$,
$\calV_j^{-}{\mathcal{V}}\calV_k^{-}=\calV_j^{-}$, so that 
${\mathcal{V}}_Y(j,k)=\underline{b}_j^T
  \calV_j^{-}\underline{b}_k=\underline{b}_j^T
  \calV_j^{-}\underline{b}_j$.  Thus, ${\mathcal{V}}_Y(j,k)={
    \mathcal{V}}_Y(j,j)$,  $j\le k, j, k=1,\ldots, K$, so that
  $Y$ has the independent increments property.

  Conversely, assume that $Y=(Y_1,\ldots ,Y_K)^T$ has the independent
  increments property.  We wish to show that there exists a vector
  $b=(b_1,\ldots ,b_K)^T$ such that $a_j=\calV_j^{-} \underline{b}_j$
  for $j=1,\ldots ,K$.  Define $q_j={\mathcal{V}}_ja_j$,
  $j=1,\ldots ,K$, in which case $a_j=\calV_j^{-}q_j$.  Note that,
  because ${\mathcal{V}}$ is of full rank, the $(j \times j)$ upper
  left submatrix of ${\mathcal{V}}$ has a unique inverse, which is the
  $(j \times j)$ upper left submatrix of $\calV_j^{-}$.  We now show
  that if $Y$, derived using nontrivial linear combinations, has the
  independent increments property, then if we take $b_j=q_{jj}$, the
  $j$th element of $q_j={\mathcal{V}}_ja_j$, then $a_j$ must equal
  $\calV_j^{-}\underline{b}_j$, which proves the converse.  As above,
  for $\ell \le \ell^{'}$,
\begin{equation}
 {\mathcal{V}}_Y(\ell,\ell^{'}) = a_{\ell}^T{\mathcal{V}}a_{\ell^{'}}
                                  = q_{\ell}^T\calV_{\ell}^{-}{\mathcal{V}}
       \calV_{\ell^{'}}^{-} q_{\ell^{'}}  = q_{\ell}^T\calV_{\ell}^{-} q_{\ell^{'}}  
                              = a_{\ell}^Tq_{\ell^{'}}. \label{eqn5}
\end{equation}
The proof is by induction.   Suppose that it has already been shown that 
that $q_{\ell}=\underline{b}_\ell$, where
$\underline{b}_\ell=(b_1,\ldots , b_{\ell},0,\ldots ,0)^T$, and
$b_{\ell}=q_{\ell\ell}$, $\ell=1,\ldots, j-1$.  Then we must 
show that $q_j=\underline{b}_j$, or $q_{j1}=b_1,\ldots
,q_{j,j-1}=b_{j-1}$.  Because $Y$ has the independent 
increments property, 
${\mathcal{V}}_Y(\ell,\ell^{'})={\mathcal{V}}_Y(\ell,\ell)$, $\ell\le
\ell^{'}, \ell,\ell^{'}=1,\ldots ,K$, which, by (\ref{eqn5}), implies
that
\begin{equation}\label{eqn6}
  a_{\ell}^Tq_{\ell^{'}}=a_{\ell}^Tq_{\ell}.
   \end{equation}
   Take $\ell=1$ and $\ell^{'}=j$.  Then (\ref{eqn6}) implies that
   $a_{11}q_{j1}=a_{11}q_{11}$, and because $a_{11}\ne 0$, it follows
   that $q_{j1}=q_{11}=b_1$.  Next, take $\ell=2$ and $\ell^{'}=j$.
   Then (\ref{eqn6}) implies that
   $a_{21}q_{j1}+a_{22}q_{j2}=a_{21}q_{21}+a_{22}q_{22}$.  We already
   showed that $q_{j1}=q_{11}$, and, by assumption. $q_{21}=b_1$.
   Thus, $ a_{21}b_1+a_{22}q_{j2}=a_{21}b_1+a_{22}q_{22}$.  Because
   $a_{22}\ne 0$, it follows that $q_{j2}=q_{22}=b_2$. Continuing in
   this fashion, we obtain 
  $q_{j\ell}=b_{\ell}$,   $\ell=1,\ldots, j-1$, and finally $q_{jj}=b_j$, completing the proof.  
\end{proof} 

\section{Implementation} \label{s:implement}

\subsection{Choice of linear combination}\label{s:genmod}

Theorem~\ref{thm1} is a very general result, demonstrating that 
linear combinations of time ordered random variables
$X=(X_1,\ldots, X_K)^T$ with covariance matrix ${\mathcal{V}}$ will
have the independent increments property if the coefficients of the
linear combination satisfy
$$a_j=\calV_j^{-} \underline{b}_j, \mbox{ where
}\underline{b}_j=(b_1,\ldots, b_j,0\ldots, 0)^T$$ for some arbitrary
vector $b=(b_1,\ldots, b_K)^T$. Consequently, to construct a linear
combination of elements of $X$ with independent increments, one can
choose any arbitrary $K$-dimensional vector of constants $b$. However,
depending on the objective, judicious choices of $b$ can be
identified.  

We now discuss considerations for the choice of $b$ when interest
focuses on a null hypothesis $H_0$ of no treatment effect to be tested
in a clinical trial potentially involving $n$ subjects at up to $K$
interim analyses.  Let $X_n = (X_{1,n}, \ldots, X_{K,n})^T$ be the
vector of suitably normalized statistics used to form the test
statistic for this hypothesis to be computed sequentially over time.
For example, if the test statistic is the score test statistic based
on a suitable log likelihood, $X_{j,n}$ is equal to $n^{-1/2}$ times
the score.  We emphasize that the test statistic computed in practice
at the $j$th interim analysis at time $t_j$ is based only on the data
accumulated through $t_j$ and does not involve $n$; we discuss this
further below.  We consider normalized statistics $X_{j,n}$ solely for
the theoretical developments that support how $b$ should be chosen.
Here, taking $n$ to grow large ($n \rightarrow \infty$) can be viewed
as allowing the staggered entry times at which subjects enroll in the
trial to become more and more dense.  The statistics can be based on
parametric, semiparametric, or nonparametric models and, under $H_0$,
generally satisfy $X_n \inD N( 0, \calV)$ as $n \rightarrow \infty$;
i.e., $X_n$ converges in distribution to a normal random vector with
mean zero and covariance matrix $\calV$, so that, for $n$ large,
$$X_n \simapprox N(0, \calV) \hspace{0.1in} \mbox{under $H_0$}.$$  
Letting $\widehat{\mathcal{V}}_n$ be a consistent estimator for
$\calV$, typically, at the $j$th interim analysis, one computes the
standardized test statistic
    $$X_{j,n} \big/\{\widehat{\mathcal{V}}_n(j,j)\}^{1/2},$$
    which thus has an approximate $N(0, 1)$ distribution under $H_0$,
    and $H_0$ is rejected if the test statistic (or absolute value
    thereof) exceeds some critical value.  In general, $X_n$ may or
    may not have the independent increments property (asymptotically).

    Ordinarily, the goal is for a test procedure to have adequate
    power to detect some alternative hypothesis of
    interest. Typically, under a fixed alternative $H_A$, common test
    procedures have the property that power converges to one as
    $n \rightarrow \infty$, so that usual asymptotic theory yields no
    insight into practical performance.  Thus, to achieve nontrivial
    results reflecting the large sample properties of a test, rather
    than consider a fixed alternative, a standard approach is to
    consider performance of the test under a sequence of so-called
    local alternatives $H_{A,n}$, which are such that the alternative
    hypothesis $H_{A,n}$ approaches the null $H_0$ as
    $n \rightarrow \infty$.  As ordinarily the case in such local
    power analysis, in the examples in Section~\ref{s:examples}, with
    fixed alternatives characterized by a quantity $\delta$ that
    equals zero under $H_0$, the alternative under $H_{A,n}$ is
    $\delta_n$, say, where $n^{1/2} \delta_n \rightarrow \tau$ as
    $n \rightarrow \infty$.  Under such local alternatives, it can be
    shown that $ X_n \inD N(\mu, \calV)$ for some $K$-dimensional
    vector $\mu = (\mu_1,\ldots,\mu_K)^T$ depending on $\tau$, where
    $\calV$ is the same covariance matrix as under $H_0$, so that
  $$X_n \simapprox N(\mu, \calV) \hspace{0.1in} \mbox{under $H_{A,n}$}.$$    
    Thus, to derive the optimal linear combination of $X_{1,n},\ldots,
    X_{j,n}$ at the $j$th interim analysis yielding the greatest power
    to detect the alternative $H_{A,n}$, consider the 
    test statistic
   $$\underline{c}_j^T
      \underline{X}_{j,n} \big/(\underline{c}_j^T {\mathcal{V}}_j
      \underline{c}_j)^{1/2},$$ where
    $\underline{c}_j=(c_{j1},\ldots, c_{jj},0,\ldots,  0)^T$.  The 
optimal such linear combination is that maximizing the 
noncentrality parameter
    \begin{equation}\label{eqn8}
\underline{c}_j^T      \underline{\mu}_j \big/(\underline{c}_j^T {\mathcal{V}}_j
      \underline{c}_j)^{1/2}, \hspace{0.15in} \underline{\mu}_j =
      (\mu_1, \ldots, \mu_j, 0, \ldots, ))^T.
  \end{equation}
  Maximizing the noncentrality parameter (\ref{eqn8}) can be
  accomplished using the Cauchy-Schwartz inequality, leading to the
  optimal choice of $\underline{c}_j$ maximizing power given by
  $\underline{c}_j^{opt}\propto{
    \mathcal{V}}_j^{-}\underline{\mu}_j$. This result suggests
  sequentially computing the optimal statistics
  \begin{equation}
  Y_{j,n}^{opt}\propto
  \underline{\mu}_j^T\calV_j^{-}\underline{X}_{j,n}.
\label{eq:Yoptjn}
  \end{equation}
  Letting $b\propto \mu$, Theorem~\ref{thm1} can be used to show that
  $Y_n^{opt}=(Y_{1,n}^{opt},\ldots, Y_{K,n}^{opt})^T$ has the independent
  increments property.

  Taking linear combinations in this way leads not only to statistics
  that yield powerful tests for detecting the alternative hypothesis
  $H_{A,n}$, but also results in test statistics that have the desired
  independent increments property. Thus, under $H_0$,
  $Y_n^{opt}=(Y_{1,n}^{opt},\ldots, Y_{K,n}^{opt})^T \simapprox
  N(0,{\mathcal{V}}_{Y^{opt}})$, where
  ${\mathcal{V}}_{Y^{opt}}(j,j)=\underline{\mu}_j^T
  \calV_j^{-}\underline{\mu}_j$, and
  ${\mathcal{V}}_{Y^{opt}}(j,k)= {\mathcal{V}}_{Y^{opt}}(j,j)$ for
  $j\le k$; and, under $H_{A,n}$,
  $Y_n^{opt} \simapprox N(\mu_{Y^{opt}}, {\mathcal{V}}_{Y^{opt}})$,
  where
  $\mu_{Y^{opt}}=\underline{\mu}_j^T \calV_j^{-}\underline{\mu}_j$.
  The test at the $j$th interim analysis is based on the (standardized)
  test statistic
  \begin{equation}
 Y_{j,n}^{opt} \big/(\underline{\mu}_j^T
 \calV_j^{-}\underline{\mu}_j)^{1/2}.
 \label{eq:Yopt}
 \end{equation}
 Because of the independent increments property, level $\alpha$
 sequential tests and boundaries can be computed readily using
 standard methods and existing software.  Note that, for standardized
 test statistics, dependence on $n$ in the numerator and denominator
 cancels, so that all necessary calculations depend only on the data
 that have accumulated at each interim analysis.

 The foregoing development assumes that a specific type of alternative
 to the null hypothesis is of interest, and the test statistic
 (\ref{eq:Yopt}) is derived to have high power against that type of
 alternative.  If in truth the distribution away from the null
 hypothesis is not consistent with such alternatives, then the test
 statistics derived above will not necessarily be optimal, but they
 will still have the independent increment property under $H_0$, and
 thus group sequential methods applied to them will preserve
 the type I error.

 \subsection{Considerations for efficient tests}
 \label{s:efftests}
 
 If the original statistics $X_{1,n},\ldots, X_{K,n}$ are indeed
 efficient for the alternative $H_{A,n}$, e.g., as for a score test
 for treatment effect parameters in a relevant parametric or
 semiparametric model, then it is well known that $X_n$ has the
 independent increments property, so that the covariance matrix
 ${\mathcal{V}}$  satisfies ${\mathcal{V}}(j,k)={\mathcal{V}}(j,j)$, $j\le k$, and the
 mean under the alternative hypothesis is proportional to the
 variance; that is, $\mu_j=\gamma {\mathcal{V}}(j,j)$.  Then, from
 (\ref{eq:Yoptjn}), 
 $$Y_{j,n}^{opt}=  \{ \gamma {\mathcal{V}}(1,1),\ldots, \gamma {\mathcal{V}}(j,j),0,\ldots, 0\} 
\calV_j^{-}\underline{X}_{j,n}.$$ However, under independent increments,
the row vector  
$\{ {\mathcal{V}}(1,1),\ldots, {\mathcal{V}}(j,j),0,\ldots, 0\}$ is the same as the $j$th row of
${\mathcal{V}}_j$. Thus,
 $$\{\gamma {\mathcal{V}}(1,1),\ldots, \gamma
   {\mathcal{V}}(j,j),0,\ldots, 0)\}{\mathcal{V}}_j^{-}=
   \gamma (0,\ldots, 0,1,0,\ldots, 0),$$
   a row vector with $j$th element $\gamma$ and zeros otherwise,
   so that $Y_{j,n}^{opt}=\gamma X_{j,n}$, $j=1,\ldots, K$, and the standardized test
   statistic as in (\ref{eq:Yopt}) based on $Y_{j,n}^{opt}$ is the same as the
   standardized test statistic based on $X_{j,n}$. That is, if we start with efficient
   tests, then the linear transformation for improvement leads us back
   to the same set of tests.

   In some settings, the treatment effect can be characterized in
   terms of a scalar parameter $\theta$ in a parametric or
   semiparametric model, with null hypothesis $H_0:\theta=0$, and a
   consistent and asymptotically normal estimator for $\theta$ is
   available that can be used to form a test statistic for $H_0$.  Let
   $\widehat{\theta}_j$ be the estimator for $\theta$ based on all the
   available data through interim time $t_j$, $j=1,\ldots, K$, and
   define
   $\underline{\widehat{\vartheta}}=(\widehat{\theta}_1,\ldots,
   \widehat{\theta}_K)^T$.  These estimators may or may not be
   efficient estimators for $\theta$ in the sense of achieving the
   Fisher information bound.  First consider the case where the
   estimators are indeed efficient, under which Scharfstein et
   al. \cite{Scharfstein} and Jennison and Turnbull \cite{Jennison}
   have shown that $X_n = n^{1/2} \underline{\widehat{\vartheta}}$ has the 
   independent increments property under $H_0$.  Importantly, in the case of
   sequentially-computed estimators, with ${\mathcal{V}}_{\theta}$ the (asymptotic)
 covariance matrix of $X_n$, the independent increments
 property takes the form
 \begin{equation}
 {\mathcal{V}}_{\theta}(j,k)={\mathcal{V}}_{\theta}(k,k), \,\,\, \mbox{ for }
 j\le k,
 \label{eq:estimatorindinc}
\end{equation}
 where ${\mathcal{V}}_{\theta}(j,k)$ is the $(j,k)$th
 element of ${\mathcal{V}}_{\theta}$.  Here, the variances
 ${\mathcal{V}}_{\theta}(j,j)$ of $X_{j,n} = n^{1/2}\widehat{\theta}_j$,
 $j = 1,\ldots,K$, decrease with $j$ as more data accrue, and the
 covariance between $n^{1/2}\widehat{\theta}_j$ and $n^{1/2}\widehat{\theta}_k$ is
 equal to the variance of the latter.

 Under a local alternative $H_{A,n}: \theta = \theta_n$, where
 $n^{1/2} \theta_n \rightarrow \tau$,
 $X_n = n ^{1/2} \underline{\widehat{\vartheta}}$ converges in
 distribution to a normal random vector with mean $\tau 1_K$ and
 covariance matrix ${\mathcal{V}}_{\theta}$, where
 $1_K =(1,\ldots, 1)^T$ is a $K$-dimensional vector of ones.  Thus, the
 optimal choice $b^{opt}$ for the vector $b$ is proportional to
 $1_K$; and, defining $\calV_{\theta,j}$ and $\calV_{\theta,j}^{-}$
 analogously to $\calV_j$ and $\calV_j^{-}$, respectively, the optimal
 choice for $Y_{j,n}^{opt}$ is given by
    $$Y_{j,n}^{opt}=1_j^T{\mathcal{V}}_{\theta,j}^{-} X_{j,n} = n^{1/2}1_j^T{\mathcal{V}}_{\theta,j}^{-}\underline{\widehat{\vartheta}}_j ,$$
    where $1_j=(1,\ldots, 1,0,\ldots, 0)^T$ is the $K$-dimensional
    vector with $j$ ones followed by $(K-j)$ zeros, and
    $\underline{\widehat{\vartheta}}_j = (\widehat{\theta}_1,\ldots,
    \widehat{\theta}_j,0,\ldots,0)^T$. Under the independent
    increments structure (\ref{eq:estimatorindinc}) for estimators, the $j$th row of
    ${\mathcal{V}}_{\theta,j}$ is equal to
    ${\mathcal{V}}_{\theta,j}(j,j)1_j^T$. Consequently,
    ${\mathcal{V}}_{\theta,j}(j,j)1_j^T
    {\mathcal{V}}_{\theta,j}^{-}=(0,\ldots ,0,1,0,\ldots, 0)$ defined
    above, so that
    $1_j^T {\mathcal{V}}_{\theta,j}^{-}=(0,\ldots
    ,0,1/{\mathcal{V}}_{\theta,j}(j,j),0,\ldots, 0)$, and
$$Y_{j,n}^{opt}=n^{1/2}\widehat{\theta}_j \big/ \mathcal{V}_{\theta,j}(j,j).$$ That
is, the optimal linear combination of
$\widehat{\theta}_1,\ldots, \widehat{\theta}_j$ leading to $Y_{j,n}^{opt}$ 
involves only $\widehat{\theta}_j$.  Thus, for $j\le k$
\begin{eqnarray*}
\cov(Y_{j,n}^{opt},Y_{k,n}^{opt})=  n \,\cov(\widehat{\theta}_j, \widehat{\theta}_k)\big/\{{
                           \mathcal{V}}_{\theta,j}(j,j) {
                           \mathcal{V}}_{\theta,k}(k,k)\} 
= {\mathcal{V}}_{\theta,k}(k,k)\big/\{{\mathcal{V}}_{\theta,j}(j,j) {
      \mathcal{V}}_{\theta,k}(k,k) \} =
  1\big/{\mathcal{V}}_{\theta,j}(j,j)=\var(Y_{j,n}^{opt}), 
\end{eqnarray*}
demonstrating independent increments. Moreover, the mean of
$Y_{j,n}^{opt}$ is equal to
$\tau\big/ {\mathcal{V}}_{\theta,j}(j,j)=\tau \var(Y_{j,n}^{opt})$,
and it is straightforward that the standardized test statistic
$Y_{j,n}^{opt}\big/ \{\var(Y_{j,n}^{opt})\}^{1/2}$ is the same as the
standardized test statistic 
$n^{1/2} \widehat{\theta}_j/ \{{\mathcal{V}}_{\theta,j}(j,j)\}^{1/2} = 
\widehat{\theta}_j\big/\{\var(\widehat{\theta}_j)\}^{1/2}$. That is,
using efficient linear combinations of efficient estimators resulting
in the independent increments property to construct a test statistic
for $H_0:\theta=0$ leads to the same result as that using the
efficient estimators alone.  In practice, an estimator for
$\var(\widehat{\theta}_j)$ would be substituted.

In the case where the estimators for $\theta$ are not efficient and
thus do not necessarily have the independent increments property, it
is still advantageous to take such linear combinations.  Specifically,
taking $b=1_K$, compute $(Y_{1,n}^{opt},\ldots, Y_{K,n}^{opt})^T$, where
$Y_{j,n}^{opt}=n^{1/2} 1_j^T{
  \mathcal{V}}_{\theta,j}^{-}\underline{\widehat{\vartheta}}_j$, and
reject the null hypothesis whenever the standardized test statistic
$$n^{1/2} 1_j^T{
  \mathcal{V}}_{\theta,j}^{-}\underline{\widehat{\vartheta}}_j\big/
\{1_j^T{ \mathcal{V}}_{\theta,j}^{-}1_j\}^{1/2}$$ exceeds some
boundary value.  In practice, an estimator for
$\mathcal{V}_{\theta,j}$ would be substituted.  By construction,
these tests have the independent increments property, so that standard
group sequential software can be used to compute the boundary values,
and will have higher power than the usual test based on the statistic
$\widehat{\theta}_j\big/\{\var(\widehat{\theta}_j)\}^{1/2}$.  This
approach is used by Van Lancker et al. \cite{Lancker} in the situation
described in Section~\ref{s:intro}.  Namely, by taking linear
combinations of sequentially-computed, potentially inefficient
covariate adjusted estimators $\widehat{\theta}_j$, $j = 1,\ldots,K$,
that do not have the independent increments property, these authors
derived a test enjoying this property that is potentially more
efficient.
 
\section{Examples}\label{s:examples}

\subsection{Preliminaries}
\label{s:prelim}

We now consider two settings involving a time to event outcome where
an existing test procedure for a relevant null hypothesis $H_0$ of no
treatment effect is based on statistics without the independent
increments property.  For each, we demonstrate how the foregoing
developments can be used to obtain modified statistics that have this
property and may target specific types of alternatives of interest.  As is
the case in many such formulations, the derivation is based on
identifying the influence function \cite{TsiatisBook}  associated with the original test
statistic under $H_0$ by representing the statistic as asymptotically
equivalent to a sum of independent and identically distributed (iid)
quantities, from which obtaining the covariance matrix of
sequentially-computed statistics is straightforward.  As such, these
examples illustrate how the proposed methodology can be implemented in
practice.

Consider a clinical trial comparing two treatments coded as 0 and 1,
where up to $n$ individuals enter the study in a staggered fashion at
iid times $E_1 \leq \cdots \leq E_n$ measured from the start of the
trial at time $t=0$.  For definiteness, suppose that the interim
analyses are conducted (if the trial has not already been stopped) at
times $t_1 < \cdots < t_K$, with $t_1>0$.  Let $Z \in \{0, 1\}$ indicate
randomized treatment assignment, and denote the event time for an
arbitrary individual, measured from study entry, by $T$.  Define the
treatment-specific time to event/survival distributions as
$S_z(u) = P(T \ge u \mid Z=z)$, $z = 0, 1$, with corresponding hazard
functions $\lambda_z(u)$, $z=0, 1$.  For simplicity, we assume that
censoring is only for administrative reasons; that is, at any study
time $t$, for an individual who enters the study at time $E$ prior to
$t$, the event time $T$ is observed if $T\le t-E$; otherwise, the
event time is censored at $t-E$. Other forms of censoring, e.g., due
to loss to follow up or drop out, can be accommodated under suitable
assumptions.  Assuming that subjects enter the trial according to a
completely random process, $E \independent (T, Z)$, where
``$\independent$'' denotes ``independent of;'' and, indexing subjects
by $i$, we take $(E_i, Z_i, T_i)$, $i = 1, \ldots, n$, to be iid.  At
interim analysis time $t$, we observe data only for individuals $i$
for whom $E_i \leq t$, and, as above, we observe $T_i$ if
$T_i \leq t-E_i$; otherwise, $T_i$ is censored at time $t - E_i$.

\subsection{Interim monitoring based on Gehan's Wilcoxon test}
\label{s:sludandwei}

In the setting of time to event outcome, a null hypothesis of central
interest is that of equality of the treatment-specific survival
distributions; i.e., $H_0: S_1(u)=S_0(u)$ for all $u \geq 0$, which is
equivalent to that of equality of the corresponding hazard functions,
$H_0:\lambda_1(u)=\lambda_0(u)$.  As in Section~\ref{s:intro}, when
the assumption of proportional hazards is likely to be violated,
analysts may wish to base a test of $H_0$ on a test statistic other
than that for the standard logrank test.  When differences in
survival are anticipated to occur early in time, Gehan's Wilcoxon
test, which may be more sensitive to such differences, is an
attractive alternative; for example, Jiang et al. \cite{NIAITP}
demonstrate this advantage in analyses of data from the National
Institute on Aging Interventions Testing Program, where the
interventions compared are expected to influence mortality prior to
middle age but not afterward.

Slud and Wei \cite{SludandWei} have shown that the statistics on which
Gehan's Wilcoxon test is based with censored survival data, computed
sequentially, do not have the independent increments property,
complicating the test's use in the context of interim monitoring.
These authors demonstrate how group sequential tests that preserve the
desired significance level can be constructed using an
$\alpha$-spending function and recursively computing multivariate
normal integrals to obtain stopping boundaries based on the asymptotic
joint distribution of the statistics under the null hypothesis.
However, given the appeal of using standard software for this purpose,
we show how the main result in Section~\ref{s:model} can be used to
obtain modified statistics that do have the independent increments
property by identifying appropriate linear combinations of the
sequentially-computed statistics while not sacrificing power.  

Under $H_0$, $S_0(u) = S_1(u) = S(u)$ and
$\lambda_1(u)=\lambda_0(u) = \lambda(u)$, say.  Define the event time
counting process $N_T(u)=I(T\le u)$ and at risk process
$\mathcal{Y}_T(u)=I(T\ge u)$, where $I(\,\cdot\,)$ denotes the indicator
function, and let $dM_T(u)=dN_T(u)-\lambda(u) \mathcal{Y}_T(u)du$ denote the
associated martingale increment.  Tarone and Ware \cite{Tarone} have
shown that Gehan's Wilcoxon test statistic with censored data can be
written equivalently as a weighted logrank test; namely, at time $t$,
the (normalized) statistic on which the test is based is given by
$$G_n(t)=n^{-1/2} \sum_{i=1}^n I(E_i\le
t)\int_0^t\frac{W_n(u,t)}{n}\{Z_i-\bar{Z}(u,t)\}dN_{T_i}(u)I(t-E_i\geq
u), \hspace{0.15in}
W_n(u,t)=\sum_{i=1}^n I(T_i\geq u,t-E_i\geq u),$$ and
$$\bar{Z}(u,t)=\frac{\sum_{i=1}^n Z_i I(T_i\geq u,t-E_i\geq u)}{\sum_{i=1}^n
  I(T_i\geq u,t-E_i\geq u)}.$$
By an algebraic identity, under $H_0$,
$$G_n(t)=n^{-1/2} \sum_{i=1}^n I(E_i\le
t)\int_0^t\frac{W_n(u,t)}{n}\{Z_i-\bar{Z}(u,t)\}dM_{T_i}(u)I(t-E_i \geq u),$$
and it can be shown that 
$$G_n(t)=n^{-1/2}\sum_{i=1}^n IF_i(t)+o_P(1),$$
where 
\begin{equation}\label{eqn8.5}
  IF_i(t)=I(E_i\le t)\int_0^tw(u,t)(Z_i-\pi)dM_{T_i}(u)I(t-E_i \geq u)
\end{equation}
is the $i$th influence function of $G_n(t)$; $w(u,t)=P(T\geq u, t-E\geq u)$ is
the limit in probability of $W_n(u,t)/n$; $\pi = P(Z=1)$ is the limit
in probability of $\bar{Z}(u,t)$, which follows because
$E \independent (Z,T)$ and $Z \independent T$ under the null
hypothesis; and $o_P(1)$ is a term that converges to zero as
$n \rightarrow \infty$.  Thus,
$$G_n(t)  \inD N\left[\,0,\var\{IF(t)\}\,\right] \hspace{0.1in} \mbox{under $H_0$},$$
and $\{G_n(s), G_n(t)\}$, $s < t$, converges in distribution under $H_0$ to
a bivariate normal random vector with mean zero and covariance matrix
 \[
  \left[ {\begin{array}{cc}
   \var\{IF(s)\} & \cov\{IF(s),IF(t)\}\\
   \cov\{IF(s),IF(t)\} &  \var\{IF(t)\}\\
  \end{array} } \right].
\]
By standard counting process martingale results, and absorbing
$I(E\le t)$ into $I(T\ge u,t-E\ge u)$ for all $u\ge 0$,
\begin{eqnarray}
  \var\{IF(t)\} &= &E \left\{\int_0^t w^2(u,t)(Z-\pi)^2I(T\ge u,t-E\ge
  u)\lambda(u)du\right\} \nonumber \\
 &= &\pi(1-\pi)\int_0^t w^2(u,t)E\{I(T\ge u,t-E\ge
                  u)\}\lambda(u)du =\pi(1-\pi)\int_0^t w^3(u,t)\lambda(u)du, \label{eqn11}
\end{eqnarray}
where, under $H_0$, (\ref{eqn11}) follows by taking the expectation
inside the integral.  Similarly, for $s\le t$,
\begin{eqnarray}
  \cov\{IF(s),IF(t)\} &= &E \left\{\int_0^s w(u,s)w(u,t)w((Z-\pi)^2I(T\ge u,s-E\ge
  u)\lambda(u)du\right\} \nonumber \\
  &= &\pi(1-\pi)\int_0^s
                        w(u,s)w(u,t)E\{I(T\ge u,s-E\ge
                  u)\}\lambda(u)du = \pi(1-\pi)\int_0^s w^2(u,s)w(u,t)\lambda(u)du. \label{eqn12}
\end{eqnarray}
From (\ref{eqn11}) and (\ref{eqn12}), in general,
$\cov\{IF(s),IF(t)\}\ne \var\{IF(s)\}$ for $s\le t$, and thus the
statistic $G_n(t)$ does not have the independent increments property.

Let $X_{j,n} = G_n(t_j)$, $j = 1,\ldots,K$, denote the original
sequential statistics, and let $\mathcal{V}_{IF}$ be the $(K \times
K)$ asymptotic covariance matrix of $(X_{1,n},\ldots,X_{j,n})^T$, 
obtained from (\ref{eqn11}) and
(\ref{eqn12}).  Gehan's Wilcoxon test statistic at the $j$th interim
analysis would be computed as
$$X_{j,n} \big/ \{ \widehat{\mathcal{V}}_{IF}(j,j)\}^{1/2},$$
where $\widehat{\mathcal{V}}_{IF}$ is an estimator for
${\mathcal{V}}_{IF}$ obtained by estimating (\ref{eqn11}) and
(\ref{eqn12}).  Namely, using standard counting process methods, and defining
$\widehat{\pi}(t)=\sum_{i=1}^n Z_i I(E_i\le t)\big/\sum_{i=1}^n
I(E_i\le t)$, $\var\{IF(t)\}$ 
in (\ref{eqn11}) can be estimated by 
\begin{eqnarray}
\widehat{\var}\{IF(t)\}=\widehat{\pi}(t)\{1-\widehat{\pi}(t)\}\int_0^t\left\{\frac{W_n(u,t)}{n}\right\}^3\frac{dN(u,t)}{W_n(u,t)}=\widehat{\pi}(t)\{1-\widehat{\pi}(t)\}\int_0^t\left\{\frac{W_n(u,t)^2}{n^3}\right\}dN(u,t), \label{eqn13}
\end{eqnarray}
where $N(u,t)=\sum_{i=1}^n \{ I(E_i\leq t) I(T_i \leq u, T_i \leq
t-E_i)\}$, and,
similarly, a consistent estimator for $\cov\{IF(s),IF(t)\}$ in
(\ref{eqn12}) is given by
  \begin{equation}\label{eeqn14}
\widehat{\cov}\{IF(s),IF(t)\}=\widehat{\pi}(t)\{1-\widehat{\pi}(t)\}\int_0^s\left\{\frac{W_n(u,s)^2}{n^3}\right\}dN(u,t).
\end{equation}

We now describe several approaches to choosing linear combinations of
the sequentially-computed statistics $X_{j,n}$ that result in
statistics that do have the independent increments property, which
involve choice of the vector of constants $b = (b_1,\ldots, b_K)^T$.
In the first two approaches, we choose $b$ to target specific
alternative hypotheses.  It is well known that the Wilcoxon test in
the case of no censoring has high power to detect log-odds
alternatives; i.e., where
\begin{equation}\label{eqn14.5}
  \log\left\{\frac{S_1(u)}{1-S_1(u)}\right\}=\log\left\{\frac{S_0(u)}{1-S_0(u)}\right\}+\delta,
  \hspace*{0.1in}\mbox{ equivalently }\hspace*{0.1in}
  S_1(t)=S_1(u, \delta) = \frac{S_0(t) \exp(\delta)}{1+S_0(t)\{\exp(\delta)-1\}},
  \end{equation}
  and $\delta$ denotes the treatment effect parameter.  We show in 
  Appendix~\ref{app:a} that, under local log-odds alternatives with treatment
  effect parameters $\delta_n$ such that
  $n^{1/2}\delta_n\rightarrow \tau$, the resulting Gehan's Wilcoxon
  statistic $G_n(t)$ converges in distribution to a normal random
  variable with variance $\var\{IF(t)\}$ (same as under the null
  hypothesis) but with mean
\begin{equation}\label{eqn14.6}
  \mu(t)=-\tau \pi(1-\pi)\int_0^tw^2(u,t)S(u)\lambda(u)du.
  \end{equation}
A consistent estimator for $\mu(t)$ in (\ref{eqn14.6})  is given by
\begin{equation}\label{eqn14.7}
\widehat{\mu}(t)=-\tau
\widehat{\pi}(t)\{1-\widehat{\pi}(t)\}\int_0^t\left\{\frac{W_n(u,t)}{n^2}\right\}\widehat{S}(u,t)dN(u,t),
\end{equation}
where
$\widehat{S}(u,t)$ is the Kaplan-Meier estimator of the survival
distribution under $H_0$ using all the available censored survival
data combined over both treatments through time $t$. The proposed test
statistic at the $j$th interim analysis is then
\begin{equation}
Y_{j,n} \big/
(\underline{\widehat{\mu}}_j^T\widehat{{\mathcal{V}}}_{IF,j}^{-}
\underline{\widehat{\mu}}_j)^{1/2}, \hspace{0.15in}
Y_{j,n}=\underline{\widehat{\mu}}_j^T\widehat{{\mathcal{V}}}_{IF,j}^{-}
\underline{X}_{j,n},
\label{eq:teststat}
\end{equation}
where, with $\widehat{\mu}(t) $ as in (\ref{eqn14.7}),
$\underline{\widehat{\mu}}_j=\{\widehat{\mu}(t_1),\ldots,
\widehat{\mu}(t_j),0,\ldots, 0\}^T$, and 
$\widehat{\mathcal{V}}_{IF,j}$ is the
$(K\times K)$ matrix with  upper left hand $(j\times j)$ submatrix
that of $\widehat{\mathcal{V}}_{IF}$ and the
remainder of the matrix zeros. The $Y_{j,n}$, $j=1,\ldots,K$,
have the independent increments property, and the test should be more
powerful against the log-odds alternative than that based on the
original Wilcoxon statistics $X_{1,n},\ldots,X_{j,n}$.

In some settings, a treatment may be expected to have a delayed
effect.  We thus consider the alternative hypothesis where
\begin{equation}\label{eqn14.56}
\lambda_1(u) = \lambda_1(u,\delta)=\lambda_0(u)I(u
\leq\mathcal{T}_{delay})+\lambda_0(u)\exp(\delta)I(u >\mathcal{T}_{delay});
\end{equation}
that is, the hazard function for treatment 1 is the same as that for
treatment 0 through some delay time $\mathcal{T}_{delay}$ and then is
proportional to that for treatment 0 by a proportionality constant
$\exp(\delta)$. We refer to this as a non-proportional hazards
alternative, with the null hypothesis being $\delta=0$.  We show in
Appendix~\ref{app:b} that, under local alternatives $\delta_n$ such that
$n^{1/2}\delta_n\rightarrow \tau$, the Gehan's Wilcoxon
statistic $G_n(t)$ converges in distribution to a normal random variable
with variance $\var\{IF(t)\}$, but mean
\begin{equation}\label{eqn14.65}
  \mu(t)=-\tau \pi(1-\pi)\int_{\mathcal{T}_{delay}}^tw^2(u,t)\lambda_0(u)du.
  \end{equation}
A consistent estimator for $\mu(t)$ in (\ref{eqn14.65}) is given by
\begin{equation}\label{eqn14.75}
\widehat{\mu}(t)=-\tau
\widehat{\pi}(t)\{1-\widehat{\pi}(t)\}\int_{\mathcal{T}_{delay}}^t\left\{\frac{W_n(u,t)}{n^2}\right\}dN(u,t).
\end{equation}
As above, the test statistic at the $j$th interim analysis is of the
form (\ref{eq:teststat}), where now
$\underline{\widehat{\mu}}_j=\{\widehat{\mu}(t_1),\ldots,
\widehat{\mu}(t_j),0,\ldots, 0\}^T$ with $\widehat{\mu}(t) $ as in
(\ref{eqn14.75}), and the $Y_{j,n}$, $j=1,\ldots,K$, have the
independent increments property, so that the test should be more
powerful against this alternative than that based on
$X_{1,n},\ldots,X_{j,n}$.  Note that if $\mathcal{T}_{delay} = 0$,
then 
(\ref{eqn14.56}) reduces to a proportional hazards alternative.

An ad hoc approach to choosing $b$ without specifying a targeted
alternative is to proceed as if the statistics were efficient with
independent increments, in which case the mean of the
statistics would be (asymptotically) proportional to the variance;
i.e., $\mu_j=E(X_{j,n}) = E\{G_n(t_j)\} \propto \var\{G_n(t_j)\}=\var\{IF(t_j)\}$,
$j = 1,\ldots,K$.  Then, as above, $b$ would be chosen to be proportional
to $[\var\{IF(t_1)\},\ldots, \var\{IF(t_K)\}]^T$.  Although the
Wilcoxon statistics are not efficient, this approach is simple to
implement and leads to the modified statistics
$$Y_{j,n}=\widehat{\underline{b}}_j^T\widehat{\mathcal{V}}_{IF,j}^{-}\underline{X}_{j,n},$$
where
$\widehat{\underline{b}}_j=[\widehat{\var}\{IF(t_1)\},\ldots
, \widehat{\var}\{IF(t_j)\},0,\ldots, 0]^T$, which have the independent
increments property.

We demonstrate the performance of these tests in a suite of simulation
studies, each involving 10,000 Monte Carlo group sequential trials.
In all studies, $S_0(u)$ follows an exponential distribution with
constant hazard rate equal to 1, $S_0(u)=\exp(-u)$. We consider three
sets of alternatives: proportional hazards, where $S_1(u)=S_1(u, \delta)$ is
exponential with constant hazard rate $\exp(\delta)$,
$S_1(u)=S_1(u, \delta) = \exp\{-u \exp(\delta)\}$; log-odds as in (\ref{eqn14.5}); and
non-proportional hazards as in (\ref{eqn14.56}) with $\mathcal{T}_{delay}=0.6$.  In all three cases,
$\delta=0$ corresponds to the null hypothesis.  In each trial,
individuals enter according to a uniform $U(0,2)$ distribution, and
individuals are randomized to treatments 0 and 1 with equal
probability, $\pi = P(Z = 1) = P(Z= 0) = 0.5$.  Each trial is monitored at
five interim analysis times $t=(1.0,1.5,2.0,2.5,3.0)$, and in all cases
the nominal significance level is 0.05.

We consider several group sequential test procedures based on Gehan's
Wilcoxon test.  The first, denoted as the unadjusted Wilcoxon test,
does not take account of the fact that the usual Gehan's Wilcoxon
statistics do not have the independent increments property and naively
computes stopping boundaries using standard methods for tests based on
independent increments.  The second, referred to as the adjusted
Wilcoxon test, is based on the usual Gehan's Wilcoxon statistics, with
boundaries computed to preserve the desired significance level using
an $\alpha$-spending function and recursive multivariate normal
integration using the \texttt{mvtnorm} package in R \cite{mvtnorm}
based on the asymptotic joint distribution of the
sequentially-computed statistics under the null hypothesis, as
proposed by Slud and Wei \cite{SludandWei}.  Four additional tests
based on the modified Wilcoxon statistics leading to independent
increments are also evaluated: the ad hoc version with the constant
vector $b$ chosen to be proportional to the variance of the
sequentially-computed Wilcoxon statistic, denoted Wilcoxon I; that
with $b$ chosen to favor log-odds alternatives according to
(\ref{eqn14.7}), Wilcoxon II; that with $b$ chosen to favor
proportional hazards alternatives according to (\ref{eqn14.75}) with
$\mathcal{T}_{delay}=0$, Wilcoxon III; and that with $b$ chosen to
favor non-proportional hazards alternatives according to
(\ref{eqn14.75}) with $\mathcal{T}_{delay}=0.6$, Wilcoxon IV.  For
comparison, we consider the logrank test, which has power targeted for
proportional hazards alternatives.

In each study, the potential sample size $n= 1000$, and we consider an
$\alpha$-spending function at each of the five interim analysis times
of $(0.05,0.1,0.4,0.7,1)\times 0.05$.  All tests are two-sided, and
the null hypothesis is rejected whenever the absolute value of the
standardized test statistics exceeds the boundary values, which for
the unadjusted Wilcoxon and Wilcoxon I-IV are computed in all cases
using the \texttt{ldbounds} package in R \cite{ldbounds1} and for the
adjusted Wilcoxon are computed as above.  For each simulation
scenario, we compute the empirical level and power of each test as
the proportion of the 10,000 trials for which the null hypothesis is
rejected and also report the Monte Carlo average of the number of
analyses conducted before stopping the trial.

Results are presented in Table~\ref{t:one}.  Under the null
hypothesis, all tests except the unadjusted Wilcoxon have empirical
type I error close to the nominal 0.05, which is not surprising given
that incorrect stopping boundaries are used for the latter.  As
expected, the logrank test has greatest power to detect a proportional
hazards alternative, followed by the Wilcoxon III, which targets this
alternative. The Wilcoxon II test, which targets log-odds
alternatives, has the highest power in this scenario, along with the
adjusted Wilcoxon test.  The greatest difference in power across tests
occurs under the non-proportional hazards alternative, where the
Wilcoxon IV test targeting such alternatives achieves considerably
higher power than the other tests, followed by the logrank test.  The
adjusted Wilcoxon test achieves similar power to Wilcoxon II across
alternatives.  The average number of analyses conducted is similar
across tests but somewhat lower for those with higher power, as expected.

 \begin{table}
   \caption{Probability of rejecting the null hypothesis and average
     number of analyses before stopping based on 10,000 Monte Carlo
     group sequential trials.  \label{t:one}}
    \centering
 \begin{tabular}{l l l l l }
   \hline
   & \textbf{Null}$^{\rm a}$ & \textbf{Prop-haz} & \textbf{Log-odds} & \textbf{Non-prop-haz}\\
   \textbf{Test}$^{\rm b}$ & $\delta=0.00$ & $\delta=0.23$ & $\delta=0.32$ & $\delta=0.47$\\
   \hline
   Wilcoxon (unadjusted) & 0.042 & 0.812 (3.62)$^{\rm c}$& 0.791 (3.43) & 0.279 (4.85)\\
   Wilcoxon (adjusted)    & 0.049 & 0.830 (3.56)                & 0.813 (3.37) & 0.301 (4.83)\\
   Wilcoxon I                  & 0.051& 0.807 (3.66)                 & 0.754 (3.53) & 0.411 (4.78)\\
   Wilcoxon II                 & 0.048 & 0.833 (3.56)                & 0.814 (3.37)  & 0.317 (4.83)\\
   Wilcoxon III                & 0.050 & 0.851 (3.54)                & 0.802 (3.38) & 0.558 (4.80)\\
   Wilcoxon IV                & 0.050 & 0.716 (3.84)                & 0.615 (3.76) & 0.812 (4.74)\\
   Logrank                     & 0.049 & 0.893 (3.31)                & 0.766 (3.45) & 0.776 (4.30)\\*[0.05in]
   RMST                         & 0.048 & 0.887 (3.32)                & 0.768 (3.45) & 0.783 (4.17)\\
   RMST I                       & 0.050 & 0.877 (3.33)                & 0.787 (3.41) &  0.662 (4.32)\\
   RMST II                       & 0.048 & 0.887 (3.31)                & 0.769 (3.44) &  0.781 (4.19)\\
   RMST III                     & 0.049 & 0.819 (3.60)                 & 0.611 (3.85) & 0.871 (3.97)\\
  \hline
 \end{tabular}
  \begin{tablenotes}
\item[$^{\rm a}$] Null: null hypothesis; Prop-haz: proportional hazards
alternative; Log-odds: log odds alternative; Non-prop-haz:
Non-proportional hazards alternative
\item[$^{\rm b}$] Tests are as defined in the text
\item[$^{\rm c}$] Monte Carlo average of number of interim analyses before stopping
in parentheses
  \end{tablenotes} 
 \end{table}
          
 Table~\ref{t:two} shows the Monte Carlo average of the standardized
 covariance matrix of the statistics used in constructing each test
 under the null hypothesis, obtained by dividing the empirical
 covariance matrix by the empirical variance of $X_{K,n}$.  As
 expected, the statistics associated with the logrank test and the
 Wilcoxon I, II, III, and IV tests all have empirical standardized
 covariance matrices consistent with independent increments; that
 associated with Gehan's Wilcoxon test does not.

\begin{table}
   \caption{Monte Carlo average empirical standardized covariance
     matrix for each test under the null hypothesis. \label{t:two}}
  \centering
  \begin{tabular}{l l l l l p{0.15in}  l l l l l p{0.15in} l l l l l}
    \hline
    \multicolumn{5}{c}{Wilcoxon (adjusted)}& &\multicolumn{5}{c}{Wilcoxon I} & &\multicolumn{5}{c}{Wilcoxon II}\\
    
    0.058 & 0.092 & 0.127 & 0.136 &  0.137  && 0.048 & 0.046 & 0.047 & 0.045 & 0.045 
                                        && 0.329 & 0.325 & 0.325 & 0.322 & 0.323\\
    0.092 & 0.240 & 0.334 & 0.367 & 0.371 && 0.046 & 0.240 & 0.243 & 0.241 &  0.243 &&
                                         0.325 & 0.565 & 0.568 & 0.565 & 0.567\\
    0.127& 0.334 & 0.651 & 0.725 & 0.735 && 0.047 & 0.243 & 0.685 & 0.685 & 0.687 &&
                                      0.325 & 0.568 & 0.825 & 0.823 & 0.825\\
    0.136 & 0.367 & 0.725 & 0.933 & 0.951 && 0.045 & 0.241 & 0.685 & 0.953 &0.954 &&
                                       0.322 & 0.565 & 0.823 & 0.963 & 0.965\\
    0.137 & 0.371& 0.735 & 0.951 & 1.000 && 0.045 & 0.243 & 0.687 & 0.954 &  1.000 &&
                                      0.323 & 0.567 & 0.825 & 0.965 & 1.000\\*[0.05in]
     \multicolumn{5}{c}{Wilcoxon III} && \multicolumn{5}{c}{Wilcoxon IV} &&\multicolumn{5}{c}{Logrank} \\
0.205 & 0.202 & 0.203 & 0.200 & 0.200 &&
0.000&  0.000 & 0.000 & 0.000 & 0.000  &&  0.216 & 0.215 & 0.216 & 0.214 & 0.213    \\
0.202 & 0.390 & 0.393 & 0.389 & 0.394 &&
0.000&  0.017 & 0.017 & 0.017 & 0.019 &&    0.215 & 0.427 & 0.428 & 0.425 & 0.427 \\
0.203 & 0.393 & 0.611 & 0.609 & 0.613 &&
0.000 & 0.017 & 0.055 & 0.055 & 0.057&&   0.216 & 0.428 & 0.679 & 0.679 & 0.680 \\
0.200 & 0.389 & 0.609 & 0.828 & 0.830 &&
0.000 & 0.017 & 0.055 & 0.233 &  0.232 &&   0.214 & 0.425 & 0.679 & 0.879 & 0.879 \\
0.200 & 0.394 & 0.613 & 0.830 & 1.000 &&
0.000 & 0.019 & 0.057 & 0.232  &1.000 &&   0.213 & 0.427 & 0.680 & 0.879 & 1.000 \\*[0.05in]

 \multicolumn{5}{c}{RMST}  &&   \multicolumn{5}{c}{RMST I}&& \multicolumn{5}{c}{RMST II}  \\

0.298 & 0.279 & 0.239 & 0.231  & 0.242 &&    0.379 & 0.373 & 0.370 & 0.370 & 0.369 && 0.211 & 0.207 & 0.204 & 0.204 &  0.203 \\
 0.279 & 0.560 & 0.500 & 0.467 & 0.479 &&   0.373 & 0.622 & 0.620 & 0.619 & 0.619 && 0.207 & 0.428 & 0.427 & 0.425 & 0.425 \\
0.239 & 0.500 & 0.739 & 0.691 & 0.692 && 0.370 & 0.620 & 0.837 & 0.835 & 0.834 &&  0.204 & 0.427 & 0.677 & 0.674 & 0.672   \\
0.231 & 0.467 & 0.691 & 0.872 & 0.864 && 0.370 & 0.619 & 0.835 & 0.961 & 0.959 &&  0.204& 0.425 & 0.674 & 0.879 & 0.874\\
 0.242 & 0.479 & 0.692 & 0.872 & 1.000 && 0.369 & 0.619 & 0.834 & 0.959 & 1.000 && 0.203 & 0.425 & 0.672 & 0.874 &  1.000     \\*[0.05in]
        \multicolumn{5}{c}{RMST III} \\
0.002 & 0.002 & 0.002 & 0.001 &  0.001 \\
0.002 & 0.135 & 0.135 & 0.133 & 0.134 \\
0.002 & 0.135 & 0.391 & 0.387 & 0.385   \\
0.001& 0.133 & 0.387 & 0.702 & 0.694\\
0.001 & 0.134 & 0.385 & 0.694 &  1.000    \\
    \hline
  \end{tabular}
\end{table}

\subsection{Monitoring using restricted mean survival time (RMST)}
 \label{ss:rmst}

 RMST is defined as $R(L) = E\{\min(T, L)\}$, for some specified time
 $L$, mean survival time restricted to time $L$, which can also be
 computed as $R(L) =\int_0^LS(u)\, du$, the area under the survival
 distribution $S(u)$ of $T$ through time $L$. In many studies, because
 of limited follow up, it is not possible to estimate $E(T)$ unless
 the support of $T$ is contained in the interval from 0 to the maximum
 follow up time.  Thus, of necessity, survival time must be restricted;
 e.g., at an interim analysis at time $t$, survival time can be
 observed only through $t$.  Given possibly censored observations on
 $T$, $R(L)$ can be estimated nonparametrically by
 $\int_0^L\widehat{S}(u) \,du$, where $\widehat{S}(u)$ is the
 Kaplan-Meier estimator of $S(u)$ based on these data.  Consequently,
 the null hypothesis of equality of treatment-specific survival
 distributions can be characterized in terms of the treatment effect
 parameter $\theta=R_1(L)-R_0(L)=\int_0^L\{S_1(u)-S_0(u)\} du$, the
 difference in treatment-specific RMST, as $H_0: \theta = 0$.  This
 parameter has intuitive appeal, as it can be viewed also as
 representing the expected years of life saved by using the better
 treatment over the restricted time interval $(0, L)$.  Moreover,
 $\theta$ can be estimated nonparametrically by
    $$\widehat{\theta}=\int_0^L\{\widehat{S}_1(u)-\widehat{S}_0(u)\}du,$$ 
where $\widehat{S}_z(u)$, $z = 0, 1$, are the treatment-specific
Kaplan-Meier estimators of $S_z(u)$, $z = 0, 1$.

Group sequential testing based on RMST has been discussed by several
authors \cite{Murray, LuTian}.  At interim analysis times
$t_1,\ldots, t_K$, one considers
$\underline{\vartheta} = (\theta_1,\ldots, \theta_K)^T$, where
\begin{equation}
  \theta_j=R_1(L_j) - R_0(L_j) = \int_0^{L_j}\{S_1(u)-S_0(u)\}du, \hspace{0.1in}j=1,\ldots,
  K, \hspace{0.1in} L_j \leq t_j.
  \label{eq:thetaj}
\end{equation}
Because of limited follow up at $t_j$, the restricted time $L_j$ must
be less than $t_j$, and $L_j$, $j=1,\ldots, K$, may be taken to
increase with increasing $t_j$ to reflect that additional data are
accrued as the study progresses.  The null hypothesis is then
$H_0:\theta_1 = \cdots =\theta_K=0$, and monitoring the study and
testing the null hypothesis entails estimating $\underline{\vartheta}$
by some estimator
$\underline{\widehat{\vartheta}} = (\widehat{\theta}_1, \ldots,
\widehat{\theta}_K)^T$ and rejecting $H_0$ at the first time the
estimated treatment effect parameter (or absolute value thereof),
suitably normalized and standardized, exceeds some critical value.  Letting
$\widehat{S}_z(u,t)$ denote the Kaplan-Meier estimator for $S_z(u)$,
$z= 0, 1$, based on all available data through time $t$, and defining
$\widehat{R}_z(t, L) = \int_0^{L} \widehat{S}_z(u,t) \,du$, $z= 0, 1$,
the obvious estimator for $\theta_j$ in (\ref{eq:thetaj}) is
\begin{equation}
  \widehat{\theta}_j= \widehat{R}_1(t_j, L_j) -
  \widehat{R}_0(t_j, L_j) = \int_0^{L_j}\{\widehat{S}_1(u,t_j)-\widehat{S}_0(u,t_j)\}du, 
\hspace{0.1in} j = 1,\ldots, K.
\label{eq:thetahatj}
\end{equation}
Taking $L_j$ not to vary with $t_j$, so that $L_j = L$ for all
$j = 1,\ldots, K$, implies that monitoring does not begin until after
$L$, i.e., $t_1 > L$, and that $\theta_j = \theta$, $j = 1,\ldots, K$.
As demonstrated by Murray and Tsiatis \cite{Murray}, under these
conditions, at each interim analysis, the common $\theta$ is estimated
using the efficient Kaplan-Meier estimators for $S_z(u)$, $z= 0, 1$,
and the sequentially-computed estimators
$\int_0^{L}\{\widehat{S}_1(u,t_j)-\widehat{S}_0(u,t_j)\}
du$, suitably normalized, have the independent increments property, so
that group sequential methods can be implemented readily with standard
software.  However, if $L_j$ and thus $\theta_j$ do vary over time as
above, then normalized versions of $\widehat{\theta}_j$ in
(\ref{eq:thetahatj}) do not have this property.  In this case, Murray
and Tsiatis \cite{Murray} and Lu and Tian \cite{LuTian} show that
group sequential tests can be constructed that preserve the desired
significance level by using an $\alpha$-spending function and
recursively evaluating multivariate normal integrals to construct
stopping boundaries based on the asymptotic joint distribution of
suitably normalized $\widehat{\theta}_j$ under the null
hypothesis.  As in Section~\ref{s:sludandwei}, to take advantage of
standard software for computing stopping boundaries, we now describe
how Theorem~\ref{thm1} can be used to derive modified statistics based
on linear combinations of sequentially-computed estimators that
have the independent increments property.

Let $n_z = \sum_{i=1}^n I(Z_i = z)$ denote the potential number of
individuals randomized to treatment $z$, $z = 0,1$, with event times $T_{z,i}$ and
entry times $E_{z,i}$, $i = 1,\ldots,n_z$.  Then $n = n_0+n_1$ is the
total potential sample size, and, $n_1/n \rightarrow \pi = P(Z=1)$ as $n
\rightarrow \infty$.  We proceed analogous to the previous section and derive the
asymptotic covariance matrix ${\mathcal{V}}_{\theta}$ of
$n^{1/2} \widehat{\underline{\vartheta}} = n^{1/2}(\widehat{\theta}_1, \ldots,
\widehat{\theta}_K)^T$.  Following previous authors
\cite{Murray,LuTian} and owing to the form $\widehat{\theta}_j$ in
(\ref{eq:thetahatj}) as the difference of treatment-specific terms, we
first find the influence functions for $\widehat{R}_z(t, L)$,
$z = 0, 1$.   The $i$th influence
function $IF_{R,z,i}(t,L) $ for $\widehat{R}_z(t, L)$ satisfies
$$n_z^{1/2} \{ \widehat{R}_z(t, L) - R_z(L)\} = n_z^{-1/2}
\sum_{i=1}^{n_z} IF_{R,z,i}(t,L) + o_P(1).$$
To derive $IF_{R,z,i}(t,L)$, note that the $i$th influence function
for the treatment-specific Kaplan-Meier estimator
$\widehat{S}_z(u, t)$ of $S_z(u)$, $IF_{z,i}(u,t)$, say, is defined
such that
$$n_z^{1/2}\{\widehat{S}_z(u,t)-S_z(u)\}=n_z^{-1/2}\sum_{i=1}^{n_z}
   IF_{z,i}(u,t)+o_P(1).$$
   Standard results for the Kaplan-Meier estimator using counting
   process methodology yield
   $$IF_{z,i}(u,t)=I(E_{z,i} \le t)S_1(u)\int_0^u
   \frac{dM_{T_{z,i}}(x)I(t-E_{z,i} \geq x)}{w_1(x,t)},$$
where $N_{T_z}(u)=I(T_z\le u)$, $\mathcal{Y}_{T_z}(u)=I(T_z\ge u)$,
$dM_{T_z}(u)=dN_{T_z}(u)-\lambda(u) \mathcal{Y}_{T_z}(u)du$,
and $w_1(x,t)=E\{I(T_{z,i}\geq x,t-E_{z,i} \geq x)\}$, and thus the $i$th
   influence function of $\widehat{R}_z(t, L)$ is given by 
  $$IF_{R,z,i}(t,L)= I(E_{z,i} \le t)\int_0^LS_1(u)\int_0^u
  \frac{dM_{T_{z,i}}(x)I(t-E_{z,i} \geq x)}{w_z(x,t)} du = \int_0^L
\frac{A_z(u,L)}{w_z(u,t)}dM_{T_{z,i}}(u) I(t-E_{z,i}\geq u),$$
where  $A_z(u,L)=\int_u^L S_z(x)\, dx$, and the second equality
follows by a change of
   variables and absorbing $I(E_{z,i} \le t)$ into $I(t-E_{z,i} \geq u)$.
   Accordingly, at $t_j$,
   \begin{equation}
n_z^{1/2} \{ \widehat{R}_z(t_j, L_j) - R_z(L_j)\}  \inD N[\, 0,
\var\{ IF_{R,z}(t_j,L_j)\}\, ],
\label{eq:rmstdiff}
 \end{equation}
   where, by standard martingale results,
   \begin{equation}
     \var\{ IF_{R,z}(t_j,L_j)\}  = E\left\{\int_0^{L_j}
       \frac{A^2_z(u,L_j)}{w^2_z(u,t_j)}I(T_{z,i}\ge u,t_j-E_{z,i}\ge
       u)\lambda_z(u)du\right\} = \left\{\int_0^{L_j}
       \frac{A^2_z(u,L_j)}{w_z(u,t_j)}\lambda_z(u)du\right\}.
     \label{eq:varj}
     \end{equation}
 A  consistent estimator for $\var\{ IF_{R,z}(t_j,L_j)\}$ in (\ref
 {eq:varj}) is given by, analogous to Nemes et al. \cite{RMSTreview},
\begin{equation}\label{eqn15}
     n_z\int_0^{L_j}
       \frac{\widehat{A}^2_z(u,t_j,L_j)}{W_{z,n}(u,t_j)\{W_{z,n}(u,t_j)-1\}}
       \, dN_z(u,t_j),
     \end{equation}
     where $\widehat{A}_z(u,t,L)=\int_u^L \widehat{S}_z(u,t)du$,
     $W_{z,n}(u,t)=\sum_{i=1}^{n_1} I(T_{z,i}\ge u, t-E_{z,i} \ge u)$,
     and $dN_z(u,t)=\sum_{i=1}^{n_1} I(T_{z,i} \leq u, t-E_{z,i} \geq  u)$.
     Similarly,
     $[\, n^{1/2}_z \{ \widehat{R}_z(t_j, L_j) - R_z(L_j)\},
     n^{1/2}_z\{ \widehat{R}_z(t_k, L_k) - R_z(L_k)\}\,]^T$,
     $j \leq k$ converges in distribution to a bivariate normal random
     vector with mean zero and covariance matrix with diagonal
     elements $\var\{ IF_{R,z}(t_j,L_j)\}$ and
     $\var\{ IF_{R,z}(t_k,L_k)\}$ and off-diagonal elements
     $\cov\{ IF_{R,z}(t_j,L_j), IF_{R,z}(t_k,L_k)\}$.  Using standard
     counting process methods,
\begin{align}
  \cov\{IF_{R,z,i}(t_j,L_j),IF_{R,z,i}(t_k,L_k)\}
  &=
      E\left\{\int_0^{L_j} \frac{A_z(u,L_j)
      A_z(u,L_k)}{w_z(u,t_j)
      w_z(u,t_k)}I(T_{z,i}\ge u,t_j-E_{1,i}\ge
    u)\lambda_z(u)\,du\right\} \nonumber \\
  &=\int_0^{L_j} \frac{A_z(u,L_j)
      A_z(u,L_k)}{w_z(u,t_k)}\lambda_z(u)\,du. \label{eq:covjk}
\end{align}
Inspection of (\ref{eq:varj}) and (\ref{eq:covjk}) shows that 
 the independent increments property (for estimators) does not hold 
 unless $L_j=L_k$.   A consistent estimator for (\ref{eq:covjk}) is given by
\begin{equation}\label{eqn16}
n_z\int_0^{L_j} \frac{\widehat{A}_z(u,t_k,L_j)
  \widehat{A}_z(u,t_k,L_k)}{W_{z,n}(u,t_k)\{W_{z,n}(u,t_k)-1\}}dN_z(u,t_k).
\end{equation}
Note that in (\ref{eqn16}) we estimate both $A_z(u,L_j)$ and
$A_z(u,L_k)$ using all the data through the larger time $t_k$.

From these results, the asymptotic distribution of $n^{1/2}
\underline{\widehat{\vartheta}}$ under $H_0$ and thus the form of
$\mathcal{V}_\theta$ can be deduced; it suffices to derive the
bivariate distribution of $n^{1/2}(\widehat{\theta}_j,
\widehat{\theta}_k)^T$.  Note that, under $H_0$, $R_1(L_j) =
R_0(L_j)$, so that 
\begin{align*}
  n^{1/2}\widehat{\theta}_j &= n^{1/2} \{  \widehat{R}_1(t_j, L_j)- \widehat{R}_0(t_j, L_j) \} \\
 &=  (n/n_1)^{1/2} n^{1/2}_1 \{\widehat{R}_1(t_j, L_j) - R_1(L_j)\} - (n/n_0)^{1/2} n^{1/2}_0 \{
   \widehat{R}_0(t_j, L_j) - R_0(L_j)\} \\
&\inD N\left[0, \pi^{-1}\var\{ IF_{R,1}(t_j,L_j)\}+(1-\pi)^{-1} \var\{ IF_{R,0}(t_j,L_j)\}\right],
\end{align*}
using (\ref{eq:rmstdiff}), Slutsky's theorem, and the fact that the
treatment-specific RMST estimators are from independent samples, and
similarly for $n^{1/2}\widehat{\theta}_k$.  Moreover, using
(\ref{eq:covjk}),
\begin{align*}
\cov( n^{1/2} \widehat{\theta}_j, n^{1/2}\widehat{\theta}_k) = 
\pi^{-1} \int_0^{L_j} \frac{A_z(u,L_j)
      A_1(u,L_k)}{w_1(u,t_k)}\lambda_1(u)\,du + (1-\pi)^{-1} \int_0^{L_j} \frac{A_0(u,L_j)
      A_0(u,L_k)}{w_0(u,t_k)}\lambda_0(u)\,du.
  \end{align*}
Thus, $n^{1/2}(\widehat{\theta}_j,\widehat{\theta}_k)^T$ converges in
distribution to a normal random vector with mean zero and covariance
matrix that follows from these results.  An estimator $\widehat{\mathcal{V}}_\theta$ for
$\mathcal{V}_\theta$ follows by substitution of (\ref{eqn15}) and 
 (\ref{eqn16}) in the foregoing expressions along with the estimator
 $\widehat{\pi} = n_1/n$.  With these substitutions, the standardized
 test statistic $n^{1/2}
 \widehat{\theta}_j\big/\{\widehat{\mathcal{V}}_\theta(j,j)\}^{1/2}$ does not depend
 $n$, $n_1$, or $n_0$ and thus depends only on the data accrued
 through $t_j$.

 We are now in a position to define modified test statistics obtained
 by choosing linear combinations of the normalized,
 sequentially-computed estimators that result in statistics with the
 independent increments property. We thus must choose the vector of
 constants $b = (b_1,\ldots,b_K)^T$ accordingly.  Let
 $X_{j,n} = n^{1/2} \widehat{\theta}_j$ denote the original,
 sequentially-computed statistics.  We demonstrate how to choose $b$ to
 target the same specific alternatives as  in Section~\ref{s:sludandwei}.  
For a given such alternative $S_1(t, \delta)$ as in (\ref{eqn14.5})
and (\ref{eqn14.56}), under local alternatives $\delta_n$,
such that $n^{1/2}\delta_n\rightarrow \tau$, $X_{j,n}$ converges in
distribution to a normal random variable with variance
$\mathcal{V}_\theta(j,j)$ and mean given by the limit of 
$n^{1/2}\int_0^{L_j}\{ S_1(u,\delta_n)-S_0(u)\}du$, which equals
$$\mu(L_j) = \tau\int_0^{L_j} \dot{S}_1(u), \hspace{0.15in} \dot{S}_1(u) =
\frac{dS_1(u,\delta)}{d\delta} \bigg|_{\delta=0}.$$
For the log-odds alternative (\ref{eqn14.5}), $\dot{S}_1(u) = S_0(u)\{1-S_0(u)\}$,
and the asymptotic mean can be estimated by
\begin{equation}\label{eqn17}
\widehat{\mu}(L_j) =   \int_0^{L_j} \widehat{S}(u,t_j)\{1-\widehat{S}(u,t_j)\}du,
  \end{equation}
  where as before $\widehat{S}(u, t)$ is the Kaplan-Meier estimator of the survival
  distribution under $H_0$ using the data through time $t$ combined
  over both treatments.   For the non-proportional hazards alternative (\ref{eqn14.56}), 
  $S_1(u,\delta)=\exp\{-\Lambda_1(u,\delta)\}$, where
  $\Lambda_1(u,\delta)$ is the cumulative hazard function
  $$\Lambda_1(u,\delta)=\int_0^u\lambda_1(x,\delta)\,dx=\Lambda_0(u)I(u<\mathcal{T}_{delay})+\{\Lambda_0(u)-\Lambda_0(\mathcal{T}_{delay})\}\exp(\delta)
  I(u\geq \mathcal{T}_{delay}),$$ and thus
  $\dot{S}_1(u) =
  -S_0(u)\{\Lambda_0(u)-\Lambda_0(\mathcal{T}_{delay})\}I(u\ge
  \mathcal{T}_{delay})$, and the asymptotic mean can be estimated by
\begin{equation}\label{eqn18}
\widehat{\mu}(L_j)   = \int_{\mathcal{T}_{delay}}^{L_j}\widehat{S}(u,t_j)\{\widehat{\Lambda}(u,t_j)-\widehat{\Lambda}(\mathcal{T}_{delay},t_j)\},
\end{equation}
where $\widehat{\Lambda}(u, t)$ is the Nelson-Aalen estimator for the
cumulative hazard function under $H_0$ using the data through time
$t$ combined over both treatments.  Proportional hazards
alternatives can be targeted by taking $\mathcal{T}_{delay} = 0$.  As
before, for the alternative of interest, the proposed test statistic
at the $j$th interim analysis is
\begin{equation}
Y_{j,n} \big/
(\underline{\widehat{\mu}}_j^T\widehat{{\mathcal{V}}}_{\theta,j}^{-}
\underline{\widehat{\mu}}_j)^{1/2}, \hspace{0.15in}
Y_{j,n}=\underline{\widehat{\mu}}_j^T\widehat{{\mathcal{V}}}_{\theta,j}^{-}
\underline{X}_{j,n},
\label{eq:teststat2}
\end{equation}
where
$\underline{\widehat{\mu}}_j=\{\widehat{\mu}(L_1),\ldots,
\widehat{\mu}(L_j),0,\ldots, 0\}^T$, and
$\widehat{\mathcal{V}}_{\theta,j}$ is the $(K\times K)$ matrix where the
upper left hand $(j\times j)$ submatrix is that of
$\widehat{\mathcal{V}}_\theta$ and the remainder of the matrix zeros.
The $Y_{j,n}$, $j=1,\ldots,K$, have the independent increments
property, and the test should be more powerful against the targeted
alternative than that based on $X_{1,n},\ldots,X_{j,n}$.

We study the performance of group sequential test procedures based on
RMST and compare to that of the procedures based on the Gehan's
Wilcoxon and logrank tests in simulations studies involving 10,000
Monte Carlo trials with generative scenarios identical to those in
Section~\ref{s:sludandwei}.  For the procedures based on RMST, which
require $L_j \leq t_j$, $j=1,\ldots,K$, at each interim analysis,
because of the potential instability of $\widehat{\theta}_j$ at the tail
of the distribution, we considered choosing $L_j = t_j, t_j-0.1$, and
$t_j-0.2$.  The latter two choices resulted in empirical type I error
closest to the nominal 0.05 level of significance with no discernible
loss of power against any of the alternatives considered.  We thus
report results with $L_j = t_j-0.2$, $j=1,\ldots,K$.

We consider several group sequential test procedures: that based on
the RMST test statistic
$n^{1/2}
\widehat{\theta}_j\big/\{\widehat{\mathcal{V}}_\theta(j,j)\}^{1/2}$,
$j=1,\ldots, K$; the modified RMST test as in (\ref{eq:teststat2})
constructed using (\ref{eqn17}) to favor log odds alternatives, which
we refer to as RMST I; the modified RMST test as in
(\ref{eq:teststat2}) constructed using (\ref{eqn18}) with
$\mathcal{T}_{delay}=0$ to favor proportional hazards alternatives,
RMST II; and the modified RMST test as in (\ref{eq:teststat2})
constructed using (\ref{eqn18}) with $\mathcal{T}_{delay}=0.6$ to
favor non-proportional hazards alternatives, RMST III.  Because the
statistic for the usual RMST test does not have the independent
increments property, boundaries are obtained using an
$\alpha$-spending function and multivariate normal integration
accomplished using the \texttt{mvtnorm} package in R \cite{mvtnorm},
analogous to the construction of boundaries for the adjusted Wilcoxon
test in Section~\ref{s:sludandwei}.  For the modified tests RMST
I--III, we used the \texttt{ldbounds} package in R \cite{ldbounds1} to
compute the boundaries.  The monitoring times, alpha spending
function, and parameters for the alternatives are identical to those
in the previous section, as are the 10,000 data sets in each case.

Empirical type I error and power under the alternatives are shown
in Table~\ref{t:one}.  All tests based on RMST achieve the nominal
level of 0.05, and, as expected, RMST I has the greatest power to
detect the log-odds alternative among them and has power only somewhat
less than that for Wilcoxon II, which also targets this alternative.
RMST III has substantially greater power for detecting the
non-proportional hazards alternative than all other tests.
Interestingly, RMST I has power approaching that of the logrank test
for the proportional hazards alternative.  Although our focus here is
not on the relative performance of these test procedures, the results
suggest that basing monitoring of treatment differences on RMST
statistics as an omnibus approach may have good properties under a
range of alternatives of interest.  The Monte Carlo averages of the
standardized empirical covariance matrices for the statistics involved
in each RMST method under the null hypothesis are shown in
Table~\ref{t:two}.  As expected, that for the unmodified  RMST statistic
is not consistent with the independent increments property while those
for RMST I-III are.  The average number of analyses conducted shows a
similar pattern as that for the Wilcoxon tests.

\section{Discussion}\label{s:discuss}

There is a vast literature on monitoring treatment differences in
randomized clinical trials using group sequential tests.  Much of the
methodology is based on the premise that the associated statistics
have the independent increments property, which allows standard
algorithms and software to be used to compute stopping boundaries.
However, this property may not hold for some tests of interest in
practice.  We have demonstrated that, regardless of whether or not the
covariance matrix of the relevant statistics has the independent
increments structure, it is always possible to find linear
combinations of them that do have this structure.  Moreover, we show
how the linear combinations can be chosen judiciously to result in
tests with high power to detect specific alternatives of interest.
Thus, modified test statistics that can be derived from our
results will have improved power, and monitoring based on them can be
implemented readily using existing group sequential software.  

The approach we have presented requires the analyst to estimate the
covariance matrix of the sequentially-computed statistics on which
the test procedure is based.  Thus, care must be taken to retain or
have the means to recreate the data collected at previous interim analyses.


\bmsection*{Acknowledgments}

This research is supported by the National Cancer Institute of the
National Institutes of Health through grant R01CA280970.  

\bmsection*{Data availability statement}

No data are used in this article.  R code implementing the simulation
studies can be found online in the Supporting Information section at
the end of this article.

\bmsection*{Conflict of interest}

The authors declare no potential conflict of interests.

\bibliography{groupsequential.bib}

\begin{thebibliography}{10}
\providecommand \doibase [0]{http://dx.doi.org/}%

\bibitem{Kim}
Kim KM, Tsiatis AA. Independent increments in group sequential tests: a review.
  {\it {SORT} - Statistics and Operations Research Transactions.}
  2020\string;44\string:223--264.

\bibitem{Pocock}
Pocock SJ. Group sequential methods in the design and analysis of clinical
  trials. {\it Biometrika.} 1977\string;64\string:191--199.

\bibitem{OBF}
O'Brien PC, Fleming TR. A multiple testing procedure for clinical trials. {\it
  Biometrics.} 1979\string;35\string:549--556.

\bibitem{LanDeMets}
Lan KKG, DeMets DL. Discrete sequential boundaries for clinical trials. {\it
  Biometrika.} 1983\string;70\string:659--663.

\bibitem{Scharfstein}
Scharfstein DO, Tsiatis AA, Robins JM. Semiparametric efficiency and its
  implication on the design and analysis of group sequential studies. {\it J Am
  Stat Assoc.} 1997\string;92\string:1342--1350.

\bibitem{Jennison}
Jennison C, Turnbull BW. Group-sequential analysis incorporating covariate
  information. {\it J Am Stat Assoc.} 1997\string;92\string:1330--1341.

\bibitem{Lancker}
Van~Lancker K, Betz J, Rosenblum M. Combining covariate adjustment with
  group-sequential information adaptive designs to increase randomized trial
  efficiency. {\it Biometrics.} 2025\string;81\string:ujaf020.

\bibitem{Murray}
Murray S, Tsiatis AA. Sequential methods for comparing years of life saved in
  the two-sample censored data problem. {\it Biometrics.}
  1999\string;55\string:1085--1092.

\bibitem{LuTian}
Lu Y, Tian L. Statistical considerations for sequential analysis of the
  restricted mean survival time for randomized clinical trials. {\it Stat
  Biopharm Res.} 2020\string;13\string:210--218.

\bibitem{Gehan}
Gehan EA. A generalized Wilcoxon test for comparing arbitrarily singly-censored
  samples. {\it Biometrika.} 1965\string;5\string:203--223.

\bibitem{SludandWei}
Slud E, Wei LJ. Two-sample repeated significance tests based on the modified
  Wilcoxon statistic. {\it J Am Stat Assoc.} 1982\string;77\string:862--868.

\bibitem{TsiatisBook}
Tsiatis AA. {\it Semiparametric Theory and Missing Data}.
\newblock New York, NY: Springer, 2006.

\bibitem{NIAITP}
Jiang N, Gelfond J, Liu Q, Strong R, Nelson JF. The Gehan test identifies
  life-extending compounds overlooked by the log-rank test in the {NIA
  Interventions Testing Program: Metformin, Enalapril}, caffeic acid phenethyl
  ester, green tea extract, and
  17-dimethylaminoethylamino-17-demethoxygeldanamycin hydrochloride. {\it
  Geroscience.} 2024\string;46\string:4533--4541.

\bibitem{Tarone}
Tarone RE, Ware J. On distribution-free tests for equality of survival
  distributions. {\it Biometrika.} 1977\string;64\string:156--160.

\bibitem{mvtnorm}
Genz A, Bretz F, Miwa T, et al. {\it mvtnorm: Multivariate normal and t
  distributions}. R Foundation for Statistical Computing; Vienna, Austria:
  2025.
\newblock Accessed June 15, 2025.

\bibitem{ldbounds1}
Casper C, Cook T, Perez OA. {\it ldbounds: {Lan-DeMets} method for group
  sequential boundaries}. R Foundation for Statistical Computing; Vienna,
  Austria:  2023.
\newblock Accessed June 15, 2025.

\bibitem{RMSTreview}
Nemes S, Bulow E, Gustavsson A. A brief overview of restricted mean survival
  time estimators and associated variances. {\it Stats.}
  2020\string;3(2)\string:107-119.

\end{thebibliography}

\appendix

\bmsection{Asymptotic mean of Gehan's Wilcoxon statistic under log-odds
  alternative}
\label{app:a}

We give a heuristic argument to show that, under local log-odds
alternatives, the asymptotic mean of the statistic used to construct Gehan's Wilcoxon is given by
(\ref{eqn14.6}). Under the log-odds alternative
(\ref{eqn14.5}), the corresponding hazard functions can be found by taking minus the
derivative of the log survival functions, which yields
$$\lambda_1(u,\delta)=\lambda_0(u)+\frac{\{\exp(\delta)-1\}\lambda_0(u)S_0(u)}{1+S_0(u)\{\exp(\delta)-1\}}.$$
Under local alternatives $\delta_n$ such that $n^{1/2}\delta_n\rightarrow \tau$,
\begin{equation}\label{ss.1}
    \lambda_1(u,\delta_n)=\lambda_0(u)\{1-\delta_nS_0(u)\}+o(\delta_n),
  \end{equation}
  where $o(\delta_n)$ is a term that is of small order in $\delta_n$;
 i.e., $o(\delta_n)/\delta_n\rightarrow 0$. By
  (\ref{eqn8.5}), the sequentially-computed statistic
  $G_n(t)$ is asymptotically equivalent to
  \begin{equation}\label{ss.2}
    n^{-1/2}\sum_{i=1}^n\int_0^tw(u,t)(Z_i-\pi)dM_{T_i}(u)I(t-E_i \geq  u)
  \end{equation}
  under the null hypothesis. Because of contiguity, the asymptotic
  equivalence is also true under the local alternatives. However, for
  local alternatives, the martingale increment process is given by
  $$dM_{T_i}(u,\delta_n)=dN_{T_i}(u)-\{\lambda_1(u,\delta_n)Z_i+\lambda_0(u)(1-Z_i)\}I(T_i\ge
  u)du,$$ which has mean zero under the sequence of local alternatives.
  We write (\ref{ss.2}) as
  \begin{eqnarray}
    &&n^{-1/2}\sum_{i=1}^n\int_0^tw(u,t)(Z_i-\pi)dM_{T_i}(u,\delta_n)I(t-E_i \geq   u) \label{ss.3}\\
    &+&n^{-1/2}\sum_{i=1}^n\int_0^tw(u,t)(Z_i-\pi)\{dM_{T_i}(u)-dM_{T_i}(u,\delta_n\}I(t-E_i \geq   u). \label{ss.4}
  \end{eqnarray}
  Under local alternatives, (\ref{ss.3}) converges to the same
  distribution as (\ref{ss.2}) does under the null hypothesis; namely
  $N\left[0,\var\{IF(t)\}\right]$.  Now
  $$dM_{T_i}(u)-dM_{T_i}(u,\delta_n)=Z_i\{\lambda_1(u,\delta_n)-\lambda_0(u)\}I(T_i\ge
  u),$$ which, using (\ref{ss.1}), equals
  $-Z_i\delta_n\lambda_0(u)S_0(u)I(T_i\ge u)+o(\delta_n)$.
  Therefore, (\ref{ss.4}) is equal to
  \begin{align*}
-n^{-1/2}\sum_{i=1}^n&\int_0^tw(u,t)(Z_i-\pi)\{Z_i\delta_n\lambda_0(u)S_0(u)I(T_i\ge
  u)+o(\delta_n) \}I(t-E_i \geq   u)du \\
 &= -n^{-1}\sum_{i=1}^n\int_0^tw(u,t)(Z_i-\pi)\{Z_in^{1/2}\delta_n\lambda_0(u)S_0(u)I(T_i\ge
    u)+n^{1/2}o(\delta_n) \}I(t-E_i \geq   u)du.
\end{align*}
Because  $n^{1/2}\delta_n\rightarrow \tau$ and $n^{1/2}o(\delta_n)\rightarrow
  0$, (\ref{ss.4}) is equal to
  \begin{equation}\label{ss.5}
  -n^{-1}\sum_{i=1}^n\int_0^tw(u,t)(Z_i-\pi)Z_i\tau\lambda_0(u)S_0(u)I(T_i\ge
  u)I(t-E_i \geq   u)du+o_P(1).
\end{equation}
Again, by contiguity, (\ref{ss.5}) converges under the local
alternatives to the same limit as it would under the null hypothesis,
which, by a simple application of the law of large numbers, is equal to
$$-E_{H_0}\left\{\int_0^tw(u,t)(Z_i-\pi)Z_i\tau\lambda_0(u)S_0(u)I(T_i\ge
u,t-E_i \geq   u)du\right\}=
-\tau\pi(1-\pi)\int_0^tw^2(u,t)\lambda_0(u)S_0(u)du,$$ which is the same as
$\mu(t)$ in (\ref{eqn14.6}).



\bmsection{Asymptotic mean of Gehan's Wilcoxon statistic under delayed
  proportional hazards alternative}
\label{app:b}

Following a similar argument as that in Appendix~\ref{app:a}, under
the delayed proportional hazards alternative (\ref{eqn14.56}) and
under 
local alternatives $\delta_n$ such that $n^{1/2}\delta_n\rightarrow \tau$,
$$dM_{T_i}(u)-dM_{T_i}(u,\delta_n)=-Z_i\delta_n\lambda_0(u)I(T_i\ge
u)+o(\delta_n).$$
Thus, (\ref{ss.4}) is equal to
\begin{align}
-n^{-1/2}\sum_{i=1}^n &\int_0^tw(u,t)(Z_i-\pi)\{Z_i\delta_n\lambda_0(u)I(u\ge
\mathcal{T}_{delay})I(T_i\ge
  u)+o(\delta_n) \}I(t-E_i \geq   u) du \nonumber \\
&= -n^{-1}\sum_{i=1}^n\int_{\mathcal{T}_{delay}}^tw(u,t)(Z_i-\pi)\{Z_in^{1/2}\delta_n\lambda_0(u)I(T_i\ge
                                                u)+n^{1/2}o(\delta_n)
                                                       \}I(t-E_i \geq   u)du \nonumber \\
&=   -n^{-1}\sum_{i=1}^n\int_{\mathcal{T}_{delay}}^tw(u,t)(Z_i-\pi)Z_i\tau\lambda_0(u)I(T_i\ge
                                                                                                    u)I(t-E_i \geq   u)du+o_P(1)
    \label{ss.6}                                                                                                
  \end{align}
  because $n^{1/2}\delta_n\rightarrow \tau$ and
  $n^{1/2}o(\delta_n)\rightarrow 0$.  Again, by contiguity,
  (\ref{ss.6}) converges under the local alternatives to the same
  limit as it would under the null hypothesis, which by the law of
  large numbers is equal to
  \begin{align*}
-E \left\{\int_{\mathcal{T}_{delay}}^tw(u,t)(Z_i-\pi)Z_i\tau\lambda_0(u)I(T_i\ge
u,t-E_i \geq   u)du\right\} = -\tau\pi(1-\pi)\int_{\mathcal{T}_{delay}}^tw^2(u,t)\lambda_0(u)du,
\end{align*}
the same as (\ref{eqn14.65}).

\end{document}